\documentclass[11pt]{amsart}
\usepackage{amsfonts, amsbsy, amsmath, amssymb, mathtools}
\usepackage{color,enumerate}
\usepackage{breqn}
\PassOptionsToPackage{hyphens}{url}\usepackage{hyperref}

\usepackage[normalem]{ulem}
\usepackage{soul}

\newtheorem{thm}{Theorem}[section]
\newtheorem{lem}[thm]{Lemma}

\newtheorem{prop}[thm]{Proposition}
\newtheorem{exmp}[thm]{Example}

\newtheorem{rmk}[thm]{Remark}

\newtheorem{thm-con}[thm]{Theorem-Conjecture}
\numberwithin{equation}{section}

\theoremstyle{definition}
\newtheorem{defn}[thm]{Definition}

\newcommand{\F}{\mathbb F}
\newcommand{\cB}{\mathcal B}

\begin{document}

\title[Boomerang uniformity of a class of power maps]{ Boomerang uniformity of a class of power maps}

\author[S. U. Hasan]{Sartaj Ul Hasan}
\address{Department of Mathematics, Indian Institute of Technology Jammu, Jammu 181221, India}
\email{sartaj.hasan@iitjammu.ac.in}

\author[M. Pal]{Mohit Pal}
\address{Department of Mathematics, Indian Institute of Technology Jammu, Jammu 181221, India}
\email{2018rma0021@iitjammu.ac.in}

\author[P.~St\u anic\u a]{Pantelimon~St\u anic\u a}
\address{Applied Mathematics Department, Naval Postgraduate School, Monterey, CA 93943, USA}
\email{pstanica@nps.edu}

\begin{abstract}
We consider the boomerang uniformity of an infinite class of (locally-APN) power maps and show that its boomerang uniformity over the finite field $\F_{2^n}$ is $2$ and $4$, when $n \equiv 0 \pmod 4$ and $n \equiv 2 \pmod 4$, respectively. As a consequence, we show that for this class of power maps, the differential uniformity is strictly greater than its boomerang uniformity.
\end{abstract}



\maketitle

\section{Introduction} \label{S1}
Let $\F_q$ be the finite field with $q=2^n$ elements, where $n$ is a positive integer. We denote {the} multiplicative group of nonzero elements of $\F_q$ by $\F_q^*$. Let $f$ be a function from {the} finite field $\F_q$ to itself. It is well-known that any function from {the} finite field $\F_q$ to itself can be uniquely represented by a polynomial in $\F_q[x]$ of degree at most $q-1.$

Substitution boxes play a very crucial role in the design of secure cryptographic primitives, such as block ciphers. The differential attack, introduced by Biham and Shamir~\cite{BS91} is one of the most efficient attacks on block ciphers. To quantify the degree of security of a substitution box, used in a block cipher, against the differential attacks, Nyberg~\cite{KN93} introduced the notion of differential uniformity. For any $\epsilon \in \F_q$, the derivative of $f$ in the direction of $\epsilon$ is given by $D_{\epsilon}f(x) = f(x+\epsilon)+f(x)$ for all $x\in \F_q.$ For any $a,b \in \F_q$, the Difference Distribution Table (DDT) entry at the point $(a,b)$ of $f$  is given by
\begin{equation}\label{ddt}
 \Delta_f(a,b) = \lvert \{ x \in \F_q \mid D_{a}f(x)=b \} \rvert,
\end{equation}
and the differential uniformity  is $\Delta_f = \max \{\Delta_f(a,b) \mid a,b \in \F_q, a\neq0\}$. When $\Delta_f=1,2$, then the function $f$ is called perfect nonlinear (PN) function, almost perfect nonlinear (APN) function, respectively. It is easy to see that over finite fields of even characteristic, $\Delta_f$ is always even and hence APN functions give the optimal resistance against the differential attack. 

Wagner~\cite{DW99} introduced a new cryptanalysis method against block ciphers, which became known as  the boomerang attack. This attack may be thought of as an extension {of} the differential attack~\cite{BS91}. In order to analyze the  boomerang attack in a better way, and analogously to the Difference Distribution Table (DDT) concerning differential attack, Cid et al.~\cite{cid} introduced the notion of Boomerang Connectivity Table (BCT). Further, to quantify the resistance of a function against the boomerang attack, Boura and Canteaut~\cite{BC18} introduced the concept of boomerang uniformity, which is  the maximum value in the BCT excluding the first row and first column. For effectively computing the entries in the BCT, Li et al.~\cite{KLi19} proposed an equivalent formulation as described below. For any $a,b \in \F_q$, the Boomerang Connectivity Table (BCT) entry of the function $f$ at point $(a,b)$, denoted by $\cB_f(a,b)$, is the number of solutions in $\F_q \times \F_q$ of the following system of equations
\begin{equation}\label{bs}
 \begin{cases}
 f(x)+f(y)=b,\\
 f(x+a)+f(y+a)=b.
 \end{cases}
\end{equation}
The boomerang uniformity of the function $f$, denoted by $\cB_f$, is given by 
$$\cB_f = \max\{ \cB_f(a,b) \mid a,b \in\F_q^* \}.$$
For {any permutation $f$}, Cid et al.~\cite[Lemma 1]{cid} showed that $\cB_f(a,b) \geq \Delta_f(a,b)$ for all $(a,b)\in \F_q \times \F_q$.  Later, Mesnager et. al~\cite{SM20} showed that it holds for non-permutation functions as well. Cid et. al~\cite[Lemma 4]{cid} showed that for APN {permutations}, the BCT is the same as the DDT, except for the first row and the first column. Thus APN {permutations} offer an optimal resistance to both differential and boomerang attacks. However, over finite fields $\F_{2^n}$ with $n$ even, which is the most interesting case in cryptography, the only known example of APN {permutation} is due to Dillon~\cite{Dil10} over $\F_{2^6}$.  The existence of APN permutations over $\F_{2^n}$, $n\geq 8$ even, is an open problem and often referred to as  the Big APN Problem. Therefore, over $\F_{2^n}$, $n$ even, the functions with differential and boomerang uniformity four offer the best (known) resistance to differential and boomerang attacks. So far, there are six classes of {permutations} over $\F_{2^n}$, $n\equiv 2 \pmod 4$ with boomerang uniformity $4$ (see \cite{BC18, KLi21, KLi19, NLi21, SM20, NLi20}). 

In this paper, we consider the boomerang uniformity of an infinite class of locally-APN (see Definition~\ref{D1}) functions $f(x)=x^{2^m-1}$ over the finite field $\F_{2^n}$, where $n=2m$ with $m>1$. In Section~\ref{S2}, we recall some results concerning the differential uniformity of $f$. Section~\ref{S3} will be devoted to the boomerang uniformity of this power map and we shall show that the power map $f$ is boomerang $2$-uniform when $n \equiv 0 \pmod 4$ (i.e. when $m$ is even) and boomerang $4$-uniform when $n \equiv 2 \pmod 4$ (i.e. when $m$ is odd), respectively.

Cid et al.~\cite{cid} (see also~\cite[Theorem 1]{SM20}) showed that for permutation functions $f$, $\cB_f \geq \Delta_f$. However, perhaps, due to lack of any explicit example in the case of non-permutations, in several follow up papers of~\cite{cid} such as~\cite{Cal21, CV20, KLi21}, the term ``permutation" was not emphasized and it has been stated that for any function~$f$, the differential uniformity is less than the boomerang uniformity. In this paper, we shall show that for non-permutations, the differential uniformity is not necessarily smaller than the boomerang uniformity. To the best of our knowledge, this is the first such example. Finally, we end with some concluding remarks in Section~\ref{S4}. 

While one might wonder if investigating non-permutation is worthy, and we believe that these questions and their answers  may reveal results of interests that do have applications to cryptography. With one exception~\cite{Dil10}, all APN functions on even dimension are non-permutations. In fact, even the known example from~\cite{Dil10} is an APN permutation that is CCZ-equivalent to the known Kim (non-permutation) APN function. 
Moreover, it is known that the boomerang uniformity is not invariant under the CCZ or even extended affine equivalence, while the differential uniformity is invariant. There are many open questions asking whether by adding a linearized polynomial (or even monomial) to an APN non-permutation function might render a permutation (surely, APN) function.
If we know exactly how the boomerang uniformity behaves under such small perturbations, then we can possibly answer some of these questions. For example, a simple consequence of our main results is that there is no permutation in the CCZ-equivalent class of $x^{2^m-1}$ over $\F_{2^{2m}}$ that has boomerang uniformity smaller than $2^m-2$.

\section{Differential Uniformity of  $x\to x^{2^m-1}$}
\label{S2}

The differential properties of the power maps of the form $x^{2^t-1}$ over $\F_{2^n}$, $1 < t <n $, have been considered in~\cite{BCC11} where authors computed DDT entries $\Delta_f(1,b)$ by determining roots of linearized polynomials of the form $x^{2^t}+bx^2+(b+1)x=0$. In fact, in~\cite{BCC11} authors introduced a new type of functions, called locally-APN functions, defined as follows.

\begin{defn}\label{D1}
 Let $f$ be a power map from $\F_{2^n}$ to itself. Then the function $f$ is said to be locally-APN if
 \[
  \Delta_f(1,b) \leq 2,~\mbox{for all}~b \in \F_{2^n} \backslash \F_2.
 \]
\end{defn}
In~\cite{BCC11} authors gave an infinite class of locally-APN functions by showing that the power map $x^{2^m-1}$ over $\F_{2^{2m}}$ is locally-APN. 

The following lemma concerning the DDT entries of the power map $x^{2^m-1}$ over $\F_{2^{2m}}$ has already been proved in~\cite[Theorem 7]{BCC11}. However, we reproduce its proof here for the sake of convenience of the readers, as it will be used in computing the BCT entries in Section~\ref{S3}.

\begin{lem}\label{L2}
 Let $f(x)=x^{2^m-1}$ be a power map defined on the finite field $\F_{2^{2m}}$. Then  $\Delta_f(1,0)= 2^m-2$, $\Delta_f(1,b)\leq 2$ for all $b \in \F_{2^{2m}}\backslash \F_2$ and
 \begin{equation*}
  \Delta_f(1,1)=
  \begin{cases}
   2 &~\mbox{if m is even,}\\
   4 &~\mbox{if m is odd}.
  \end{cases}
 \end{equation*}

\end{lem}
\begin{proof}
For any $b \in \F_{2^{2m}}$, consider the DDT entry at point $(1,b)$, which is given by the number of solutions in $\F_{2^{2m}}$ of the following equation
\begin{equation}\label{Deq}
 (x+1)^{2^m-1}+x^{2^m-1} = b.
\end{equation}
We shall now split the analysis to find  the number of solutions of the above equation in the following cases.

\textbf{Case 1.} Let $b=0$. It is easy to observe that $x=0,1$ are not  solutions of the above Equation~\eqref{Deq}. For $x \neq 0,1$, Equation~\eqref{Deq} reduces to 
 \begin{equation*}
 \left(\frac{x+1}{x}\right)^{2^m-1}= 1. 
 \end{equation*}
If we let $y = 1+x^{-1}$, then the above equation reduces to $y^{2^m-1}=1$. Since $\gcd(2^m-1, 2^{2m}-1) = 2^m-1$, this equation has exactly $2^m-2$ solutions in $\F_{2^{2m}} \backslash \F_2$ and hence $\Delta_f(1,0)= 2^m-2$.

\textbf{Case 2.} Let $b=1$. Notice that in this case $x=0$ and $x=1$ are solutions of Equation~\eqref{Deq}. For $x\neq 0,1$, Equation~\eqref{Deq} is equivalent to
\begin{equation*}
 \begin{split}
  \frac{x^{2^m}+1}{x+1}+\frac{x^{2^m}}{x} = 1
  \iff x^{2^m}+x^2 = 0.
 \end{split}
\end{equation*}
With $x^2 = y$, the above equation becomes
\begin{equation}\label{b1}
 y(y^{2^{m-1}-1}+1) = 0.
\end{equation}
Notice that when $m>1$ is odd then $\gcd(m-1, 2m) = 2$ and the above equation~\eqref{b1} can have at most $4$ solutions, namely $0,1,\omega, \omega^2$, where $\omega$ is a primitive cubic root of unity. Hence $\Delta_f(1,1)=4$. When $m>1$ is even then $\gcd(m-1, 2m) = 1$, thus $0,1$ are only solutions of the equation~\eqref{b1}. Hence in this case $\Delta_f(1,1)=2$.

\textbf{Case 3.} Let $b \in \F_{2^{2m}} \backslash \F_2$. It is easy to see that in this case $x=0$ and $x=1$ are not solutions of Equation~\eqref{Deq}. Therefore, the DDT entry at $(1,b)$ is the number of solutions in $\F_{2^{2m}} \backslash \F_2$ of the following equivalent equation
\begin{equation}
\label{due1}
 x^{2^m}+bx^2+(b+1)x=0.
\end{equation}
Now, raising the above equation to the power $2^m$, we have 
\begin{equation}
\label{due2}
 x^{2^{2m}}+b^{2^m}x^{2^{m+1}}+(b^{2^m}+1)x^{2^m}=0.
\end{equation}
Combining~\eqref{due1} and~\eqref{due2}, we have
\begin{equation*}
 \begin{split}
   b^{2^m+2}x^4+(b^{2^m+2}+b^{2^m+1}+b^{2^m}+b)x^2
   +(b^{2^m+1}+b^{2^m}+b)x=0.
 \end{split}
\end{equation*}
We note that the above equation can have at most $4$ solutions in $\F_{2^{2m}}$, two of which are $0$ and $1$ and thus it can have at most two solutions in $\F_{2^{2m}} \backslash \F_2$. Therefore for $b \in \F_{2^{2m}} \backslash \F_2$, $\Delta_f(1,b) \leq 2$. This completes the proof.
\end{proof}

\section{Boomerang uniformity of  $x\to x^{2^m-1}$}\label{S3}

In this section, we shall discuss the boomerang uniformity of the locally-APN functions given in the previous section. The boomerang uniformity of the power maps of the type $x^{2^t-1}$ over $\mathbb F_{2^n}$ has been considered in~\cite{ZH19}, where the authors give  bounds on the boomerang uniformity in terms of the differential uniformity under the condition $\gcd(n,t)=1$ and also show that  the power permutation $x^7$  has boomerang uniformity $10$ over $\F_{2^n}$, where $n\geq 8$ is even and $\gcd(3, n) = 1$.  The following theorem gives the boomerang uniformity of the power map~$f(x)=x^{2^m-1}$ over $\F_{2^{2m}}$ where $m>1$ is odd.  

\begin{thm}\label{T1}
Let $f(x)=x^{2^m-1}$, $m>1$ odd, be a power map from the finite field $\F_{2^{2m}}$ to itself. Then, the boomerang uniformity of $f$ is~$4$.
\end{thm}
\begin{proof}
Recall that for any $b \in \F_q^*$, $q=2^{2m}$, the BCT entry $\cB_f(1,b)$ at point $(1,b)$ of  $f$, is given by the number of solutions in $\F_q \times \F_q$ of the following system of equations
 \begin{equation}\label{pbs}
  \begin{cases}
   x^{2^m-1}+y^{2^m-1}=b,\\
   (x+1)^{2^m-1}+(y+1)^{2^m-1}=b.
  \end{cases}
 \end{equation}
Notice that the above system~\eqref{pbs} cannot have solutions of the form $(x_1,y_1)$ with $x_1=y_1$ as $b \neq 0$. Also it is easy to observe that if $(x_1,y_1)$ is a solution of the above system~\eqref{pbs}, then so are $(y_1,x_1), (x_1+1,y_1+1)$ and $(y_1+1,x_1+1)$. We shall split the analysis of the solutions of the system~\eqref{pbs} in the following five cases.

\textbf{Case 1.} Let $x=0$. In this case, the system~\eqref{pbs} reduces to  
\begin{equation}\label{eqx0}
  \begin{cases}
   y^{2^m-1}=b,\\
   (y+1)^{2^m-1}+y^{2^m-1}=1.
  \end{cases}
 \end{equation}
From Lemma~\ref{L2}, we know that the second equation of the above system has four solutions, namely $y= 0,1,\omega$ and $\omega^2$. Also since $m$ is odd, we have $2^m-1 \equiv 1 \pmod 3$. Since $b\neq0$, $y=0$ cannot be a solution of the system~\eqref{eqx0} and $y= 1,\omega$ and $\omega^2$ will be a solution of the system~\eqref{eqx0} when $b=1,\omega$ and $\omega^2$, respectively. Equivalently, when $b=1,\omega, \omega^2$ then $(0,1), (0,\omega), (0,\omega^2)$ are solutions of the system~\eqref{pbs}, respectively. When $b \in \F_q \backslash \F_{2^2}$ then there is no solution of the system~\eqref{pbs} of the form $(0,y)$.

\textbf{Case 2.} Let $x=1$. In this case, the system~\eqref{pbs} reduces to  
\begin{equation}\label{eqx1}
  \begin{cases}
   y^{2^m-1}=b+1,\\
   (y+1)^{2^m-1}+y^{2^m-1}=1.
  \end{cases}
 \end{equation}
 Similar to the previous case, the second equation of the above system~\eqref{eqx1} has four solutions, namely $y= 0,1,\omega$ and $\omega^2$. Since $b\neq0$, $y=1$ cannot be a solution of~\eqref{eqx1} and $y= 0,\omega$ and $\omega^2$ will be a solution of~\eqref{eqx1}, when $b=1,\omega^2$ and $\omega$, respectively. Equivalently, when $b=1,\omega, \omega^2$ then $(1,0), (1,\omega^2), (1,\omega)$ are solutions of the system~\eqref{pbs}, respectively. When $b \in \F_q \backslash \F_{2^2}$ then there is no solution of the system~\eqref{pbs} of the form $(1,y)$.

\textbf{Case 3.} Let $y=0$. Since the system~\eqref{pbs} is symmetric in the variables $x$ and $y$, this case directly follows from Case 1. That is, when $b=1,\omega, \omega^2$ then $(1,0), (\omega,0), (\omega^2,0)$ are solutions of the system~\eqref{pbs}, respectively. When $b \in \F_q \backslash \F_{2^2}$ then there is no solution for~\eqref{pbs} of the form~$(x,0)$. 

\textbf{Case 4.} Let $y=1$. This case directly follows from Case 2. That is, when $b=1,\omega, \omega^2$ then $(0,1), (\omega^2,1), (\omega,1)$ are solutions of the system~\eqref{pbs}, respectively. When $b \in \F_q \backslash \F_{2^2}$ then there is no solution for~\eqref{pbs} of the form~$(x,1)$.

\textbf{Case 5.}  Let $x,y \neq 0,1$. In this case system~\eqref{pbs} reduces to 
 \begin{equation}
 \label{eq:xy1}
  \begin{cases}
   x^{2^m}y+xy^{2^m}=bxy,\\
   (x+y)^{2^m}+(b+1)(x+y)+b=0.
  \end{cases}
 \end{equation}
 Let $y=x+z$.  Then, the above system becomes 
 \begin{equation}\label{pbs3}
  \begin{cases}
   x^{2^m}z+xz^{2^m}=bx(x+z),\\
   z^{2^m}+(b+1)z+b=0.
  \end{cases}
 \end{equation}
 Now, raising the second equation of the above system to the power $2^m$, we have
 \begin{equation}
 \label{pbs4}
   (b^{2^m}+1)z^{2^m}+z+b^{2^m} =0.
 \end{equation}
 Combining the second equation of~\eqref{pbs3} and Equation~\eqref{pbs4}, we obtain
 \begin{equation}
 \label{pbs5}
    ((b+1)^{2^m+1}+1)(z+1) =0.
 \end{equation}
 Therefore, the system~\eqref{pbs3} reduces to
 \begin{equation}
 \label{pbs6}
  \begin{cases}
   x^{2^m}z+xz^{2^m}=bx(x+z),\\
   ((b+1)^{2^m+1}+1)(z+1) =0.
  \end{cases}
 \end{equation} 
Now, we shall consider following two cases
 
\textbf{Subcase 5.1.} Let $(b+1)^{2^m+1} \neq 1$. In this case, the first equation of~\eqref{pbs6} reduces to 
\[
  x^{2^m}+bx^2+(b+1)x=0,
\]
which is equivalent to 
\begin{equation}
\label{pbs7}
 b^{2^m+2}x^4+(b^{2^m+2}+b^{2^m+1}+b^{2^m}+b)x^2
 +(b^{2^m+1}+b^{2^m}+b)x=0.
\end{equation}
When $b=1$,  the above equation becomes $x^4+x =0$, which has four solutions $x=0,1,\omega, \omega^2$. Since we assumed $x,y \neq0$, the only solutions we consider are $x=\omega$ and $\omega^2$. Thus for $b=1$, $(\omega, \omega^2)$ and $(\omega^2, \omega)$ are solutions of the system~\eqref{pbs3}. When $b \in \F_q \backslash \F_2$ with $(b+1)^{2^m+1}\neq 1$, by Lemma~\ref{L2}, Equation~\eqref{pbs7} can have at most two solutions. 

\textbf{Subcase 5.2.} Let $(b+1)^{2^m+1} = 1$. It is more convenient, now, to work with~\eqref{eq:xy1}. We then raise the first equation of the system~\eqref{eq:xy1} to the $2^m$-th power obtaining
\[
x^{2^{2m}} y^{2^m}+y^{2^{2m}} x^{2^m}=b^{2^m} x^{2^m}y^{2^m},
\]
which is equivalent to 
\[
x y^{2^m}+y x^{2^m}=b^{2^m} x^{2^m}y^{2^m}.
\]
Combining this with the first equation of~\eqref{eq:xy1}, we infer that
\[
b^{2^m} x^{2^m}y^{2^m}=bxy,
\]
and so, $bxy=\alpha\in\F_{2^m}^*$.
Using $y=\frac{\alpha}{bx}$ in the first equation of~\eqref{eq:xy1}, we obtain 
\begin{equation}
\label{eq:x1}
x^{2^m-1}\frac{1}{b}+x^{1-2^m} \frac{1}{b^{2^m}}=1.
\end{equation}
Label $T=x^{2^m-1}$. Then the above equation reduces to 
 \begin{equation*}
 \begin{split}
  \frac{T}{b}+ \frac{T^{-1}}{b^{2^m}} &= 1\\
  \iff   \frac{T^2}{b}+ \frac{1}{b^{2^m}} &= T\\
  \iff   T^2b^{2^m}+b &= Tb^{2^m+1}.
  \end{split}
 \end{equation*}
Since, $(b+1)^{2^m+1} = 1$, by expansion, we get $b^{2^m+1}+b^{2^m}+b=0$, and so, $b^{2^m+1}=b^{2^m}+b$. The previous equation becomes 
\begin{equation*}
  T^2b^{2^m}+b = Tb^{2^m} +Tb\iff
  (Tb^{2^m}+b)(T+1) =0.
 \end{equation*}
 If $T=1$, then $x\in\F_{2^m}$ and so, $by\in\F_{2^m}$. Taking this back into~\eqref{eq:xy1}, we then obtain
 \begin{equation*}
  \begin{cases}
   y^{2^m}+(b+1)y=0,\\
   y^{2^m}+(b+1)y=(x+1)b,
  \end{cases}
 \end{equation*}
which is inconsistent with $x\neq 1$ and $b \in \F_q^*.$ If $Tb^{2^m}+b=0$, then we have 
\begin{equation*}
 \begin{split}
  Tb^{2^m}+b &=0 \\
  \iff Tb^{2^m-1}+1 &=0\\
  \iff (bx)^{2^m-1}&=1.
 \end{split}
\end{equation*}
Therefore $bx \in \F_{2^m}$ and hence $\frac{\alpha}{bx}= y \in \F_{2^m}$. Taking this back into~\eqref{eq:xy1}, we then obtain
\begin{equation*}
  \begin{cases}
   x^{2^m}+(b+1)x=0,\\
   x^{2^m}+(b+1)x=(y+1)b,
  \end{cases}
 \end{equation*}
 which is inconsistent with $y \neq 1$ and $b \in \F_q^*.$ This completes the proof.
\end{proof}

 \begin{exmp}
As an example, we checked by SageMath that the differential uniformity of the non-permutation power map $x^7$ over $\mathbb F_{2^6}$ is $6$, whereas its boomerang uniformity is $4$.
 \end{exmp}
 
The following theorem gives the boomerang uniformity of the power map $f(x)=x^{2^m-1}$ over $\F_{2^{2m}}$ , where  $m>1$ is even. 

\begin{thm}
\label{T2}
Let $f(x)=x^{2^m-1}$, $m>1$ even, be a power map from the finite field $\F_{2^{2m}}$ to itself. Then, the boomerang uniformity of $f$ is~$2$.
\end{thm}
\begin{proof} Following similar arguments as in the proof of Theorem~\ref{T1}, it is straightforward to see that when $b=1$, $(0,1)$ and $(1,0)$ are the only solutions of the system~\eqref{pbs} with either of the coordinates $x,y$ being $0$ or $1$. On the other hand, when $b \in \F_{2^{2m}} \backslash \F_{2}$, there is no solution of the system~\eqref{pbs} with either of the coordinates $x,y \in \{0,1\}$. 

We now consider the case when $x,y \neq 0,1$. In this case, the system~\eqref{pbs} reduces to 
 \begin{equation}
 \label{BE:xy1}
  \begin{cases}
   x^{2^m}y+xy^{2^m}=bxy,\\
   (x+y)^{2^m}+(1+b)(x+y)+b=0.
  \end{cases}
 \end{equation}
 Let $y=x+z$. Now, raising the second equation of the above system to the power $2^m$ and adding it to the second equation of the above system, we have  
 \begin{equation} \label{BE6}
  \begin{cases}
   x^{2^m}z+xz^{2^m}=bx(x+z),\\
   ((b+1)^{2^m+1}+1)(z+1) =0.
  \end{cases}
 \end{equation} 
Now, we shall consider the following two cases.
 
\textbf{Case 1.} Let $(b+1)^{2^m+1} \neq 1$. In this case, the system~\eqref{BE6} reduces to 
\[
  x^{2^m}+bx^2+(b+1)x=0,
\]
which is equivalent to 
\begin{equation}\label{BE7}
 b^{2^m+2}x^4+(b^{2^m+2}+b^{2^m+1}+b^{2^m}+b)x^2
 +(b^{2^m+1}+b^{2^m}+b)x=0.
\end{equation}
When $b=1$, the above equation becomes $x^4+x =0$, which has two solutions $x=0,1$, as $m$ is even. Since we assumed $x,y \neq0,1$, we do not get any solution of the system~\eqref{BE7} in this case. When $b \in \F_{2^{2m}} \backslash \F_2$ with $(b+1)^{2^m+1}\neq 1$, by Lemma~\ref{L2}, Equation~\eqref{BE7} can have at most two solutions. 

\textbf{Case 2.} Let $(b+1)^{2^m+1} = 1$, the argument is similar to Subcase 5.2 of Theorem~\ref{T1} and in this case the system~\eqref{pbs} will have no solution.
\end{proof}

 \begin{exmp}
The differential uniformity of the non-permutation power map $x^{15}$ over $\mathbb F_{2^8}$ is $14$, whereas its boomerang uniformity is $2$. 
 \end{exmp}

In the following we shall focus on APN functions. First, we  recall two lemmas which give a connection between the DDT and BCT entries of arbitrary permutation functions, respectively,  APN permutations.

\begin{lem}\cite[Lemma 1]{cid} \label{LL1}
 If $f$ is  a permutation function on $\F_q$, then $\Delta_f(a,b) \leq \cB_f(a,b)$, for all $(a,b)\in \F_q \times \F_q.$
\end{lem}

\begin{lem}\cite[Lemma 4]{cid}\label{LL2}
 For any permutation with $2$-uniform DDT, the BCT entries equal the DDT entries, except for the first row and the first column.
\end{lem}

From Lemma~\ref{LL1} and Lemma~\ref{LL2}, we can deduce (surely, known) that APN permutations have boomerang uniformity $2$. Of course, it is rather easy to see that the converse is also true.

\begin{prop}
 Let $f$ be a permutation on $\F_q$. If the boomerang uniformity of $f$ is $2$, then it is APN.
\end{prop}
\begin{proof}
 Since the boomerang uniformity of the permutation $f$ is $2$, we have
 \[
  \Delta_f(a,b) \leq \cB_f(a,b) \leq 2,
 \]
for all $(a,b)\in \F_q^* \times \F_q^*.$ Also notice that for any $a \in \F_q^*$,
\begin{equation*}
 \begin{split}
  \Delta_f(a,0) &= \lvert \{ x\in \F_q \mid f(x+a)+f(x)=0\} \rvert \\
  &=0.
 \end{split}
\end{equation*}
Thus $\Delta_f(a,b) \leq 2$ for all $(a,b) \in \F_q^* \times \F_q$ and hence $f$ is APN.
\end{proof}
By providing an extension of the boomerang uniformity to the case of arbitrary functions, Mesnager et. al~\cite{SM20} observed that for any arbitrary  function $f$, $\Delta_f(a,b) \leq \cB_f(a,b)$, for all $(a,b)\in \F_q \times \F_q.$ In the following, we shall show that the boomerang uniformity of any arbitrary APN function $f$ is $2$.

\begin{thm}\label{P2}
 Let $f$ be an arbitrary APN function over $\F_q$. Then the boomerang uniformity of $f$ is $2$.
\end{thm}
\begin{proof}
 Recall that the BCT entry $\cB_f(a,b)$ of $f$ at ponint $(a,b)$ is given by
 \begin{equation*}
  \cB_f(a,b) = \left \lvert \left\{ (x,y)\in \F_q \times \F_q \mid
  \begin{cases}
  f(x)+f(y)=b \\
  f(x+a)+f(y+a) =b
  \end{cases}
  \right \}  \right \rvert .
 \end{equation*}
Let $y=x+\gamma$, then the above equation becomes
\begin{equation*}
  \begin{split}
  \cB_f(a,b) &= \left \lvert \left \{ (x,\gamma)\in \F_q \times \F_q \mid 
  \begin{cases}
  f(x+\gamma)+f(x)=b \\
  f(x+a+\gamma)+f(x+a)=b
  \end{cases}
  \right \} \right \rvert \\
  &= \sum_{\gamma \in \F_q} \left \lvert \left \{ x\in \F_q \mid 
  \begin{cases}
  f(x+\gamma)+f(x)=b \\
  f(x+a+\gamma)+f(x+a)=b
  \end{cases}
  \right \} \right \rvert.
  \end{split}
 \end{equation*}
 Also observe that for any $(a,b)\in \F_q^* \times \F_q^*$, we have 
 \begin{equation*}
  \cB_f(a,b) = \sum_{\gamma \in \F_q^*} \left \lvert \left \{ x\in \F_q \mid 
  \begin{cases}
  f(x+\gamma)+f(x)=b \\
  f(x+a+\gamma)+f(x+a)=b
  \end{cases}
  \right \} \right \rvert.
 \end{equation*}
 Since $f$ is an APN function, therefore, for any $(a,b)\in \F_q^* \times \F_q^*$, the quantity under summation will contribute only if $a=\gamma$. Now for $\gamma=a$, $\cB_f(a,b)=\Delta_f(a,b) \leq 2$ and hence the boomerang uniformity of $f$ is $2$.
\end{proof}

\begin{rmk}
 The converse of the above Theorem~\textup{\ref{P2}} is not necessarily true and  counterexamples can be easily constructed.  For instance, the boomerang uniformity of the power map $x^{15}$ over $\F_{2^8}$ is $2$, though it is not an APN function. 
\end{rmk}

\section{concluding remarks}\label{S4}
In this note we compute the boomerang uniformity of the power map $x^{2^m-1}$ over $\F_{2^{2m}}$. As an immediate consequence, we find that the differential uniformity is not necessarily smaller than the boomerang uniformity (for non-permutations), as it was previously shown for permutations. The presented class is  not just an isolated example. Our computer programs reveal quickly some other like-functions, for instance, $x\mapsto x^{45}$ on $\F_{2^8}$, which is locally-APN, has differential uniformity~$14$, and boomerang uniformity~$2$. We could not extrapolate an infinite class out of all these examples, though.
 It would be interesting to construct some new (infinite) classes of functions for which the boomerang uniformity is strictly smaller than the differential uniformity. 

\section*{Acknowledgements}
We would like to express our sincere appreciation to the editors for handling our paper and to the reviewers for their careful reading, beneficial comments and constructive suggestions. 

The research of Sartaj Ul Hasan is partially supported by MATRICS grant MTR/2019/000744 from the Science and Engineering Research Board, Government of India.

\end{document}